\documentclass[11pt, A4]{article}
 
\usepackage{amssymb}
\usepackage{amsmath}
\usepackage{amsthm}
\usepackage[ruled,vlined]{algorithm2e}
\usepackage[dvipsnames]{xcolor}
\usepackage{booktabs}
\usepackage{tikz}
\usepackage{mathtools}
\usetikzlibrary{shapes}
\usepackage{todonotes}
\usepackage{graphicx}
\usepackage{url}

\newtheorem{theorem}{Theorem}

\newtheorem{lemma}[theorem]{Lemma}
\newtheorem{corollary}[theorem]{Corollary}
\newtheorem{remark}[theorem]{Remark}
\newtheorem{definition}[theorem]{Definition}
\newtheorem{example}[theorem]{Example}

\usepackage{mathtools}
\usepackage{bbm}
\usepackage{todonotes}
\usepackage{thm-restate}

\newcommand{\Pa}{\mathcal{P}}

\DeclareMathOperator*{\argmax}{arg\,max}

\newtheorem{observation}[theorem]{Observation}
 
\title{Network Investment Games with Wardrop Followers}
\author{Daniel Schmand\thanks{Goethe University Frankfurt, Germany. \texttt{schmand@em.uni-frankfurt.de}.}
 \and Alexander Skopalik\thanks{University of Twente, Netherlands. \texttt{a.skopalik@utwente.nl}.}
\and Marc Schr\"{o}der\thanks{RWTH Aachen University, Germany.\texttt{marc.schroeder@oms.rwth-aachen.de}.}}

\begin{document}
\maketitle

\begin{abstract}
We study a two-sided network investment game consisting of two sets of players, called providers and users. The game is set in two stages. In the first stage, providers aim to maximize their profit by investing in bandwidth of cloud computing services. The investments of the providers yield a set of usable services for the users. In the second stage, each user wants to process a task and therefore selects a bundle of services so as to minimize the total processing time. We assume the total processing time to be separable over the chosen services and the processing time of each service to depend on the utilization of the service and the installed bandwidth. We provide insights on how competition between providers affects the total costs of the users and show that every game on a series-parallel graph can be reduced to an equivalent single edge game when analyzing the set of subgame perfect Nash equilibria.
\end{abstract}

\section{Introduction}
With the  increasing availability  of  a fast and reliable internet connection, the demand for and the importance of cloud-computing services is heavily growing. That is why large IT-companies such as  Amazon, Google or Microsoft have been investing massively in the improvement of their cloud services in recent years. New high-speed cables are being placed, high performance hardware has been installed, new software has been developed, and machine numbers have been scaled up  \cite{hecht2016bottle}. Interestingly, prices for these services are typically charged per usage time (see for example, Amazon Web Services \cite{amazon}, Microsoft Azure \cite{azure} and Google Cloud \cite{google}). This might a-priory lead to a situation, where a non-optimal infrastructure takes longer to process the users' tasks, and thus even leads to a higher income for the corresponding provider. This paper addresses this issue by means of a theoretical model and provides some results on how competition between providers helps to prevent this behavior of providers.

We study the problem by means of a two-stage game, which we call a network investment game. In the first stage, a set of providers invests in bandwidths of services that are used by users in the second stage to process their task. These services could be different types of on-demand cloud computing, storage, data-base, security or media services. We assume that these services are either complements or substitutes to each other. This implies that the set of feasible bundles of services for the users can be represented by means of source-sink paths in a series-parallel network. Two services that are complementary can be modeled as two edges that are connected in series, whereas two services that are substitutes to each other can be modeled by means of two parallel edges. A similar idea has been used by Correa et al.\ in \cite{correa2014sensitivity}. We call this underlying graph the technology graph. However, before a user can use a particular service, there must be a provider that has invested in bandwidth for that particular service. In other words, the investments of the providers induce a subgraph, the set of usable services, of the original technology graph that can be used by the users to process their task. Each user selects a bundle of services that is needed to perform the user's task at minimum total processing time. The total processing time of a service depends on the quantity of users that chooses the particular service and the installed bandwidth of the providers. We assume that each user has a given willingness to pay for taking part in the game and is only present if the total costs do not exceed this value. Given that users impose externalities on each other, we assume that users in the second stage choose their bundle of services according to a Wardrop equilibrium \cite{wardrop1952some} in the corresponding graph of usable services. 

Given that prices for services are typically charged per usage time, we use the users' costs as a proxy for the total revenue of the providers investing in a service. This total revenue is shared proportionally among those providers investing in the service. Proportional sharing is used in many different settings, e.g.,\ in cooperative game theory \cite{Moulin1987} and in the context of congestion games \cite{milchtaich1996congestion}. We assume that the providers' aim is to maximize profit, which depends on the amount of users using their services and the total amount invested. Intuitively, little investment in a service implies a high processing time and thus a high price paid by users, but only a small fraction of users will choose the service. Too much investment might lead to a large fraction of users, but only for a low price.

\subsection{Contribution}
We are interested in how competition between selfish providers affects the choices made by the users in the network investment game. More formally, we study the existence and uniqueness of subgame perfect Nash equilibria in which in the first stage, providers maximize profit and in the second stage, users choose a Wardrop equilibrium given the strategies of the providers in the first stage. We then consider the performance of equilibria in terms of the price of anarchy \cite{koutsoupias1999worst}. The price of anarchy measures the difference in performance between the worst subgame perfect Nash equilibrium and the social optimum. In fact, for all instances we consider, the performance of all subgame perfect Nash equilibria is the same and we are able to provide an exact characterization of this performance.

Our main result, Theorem \ref{thm:sho}, provides a simplification of the structure of equilibrium investments for series-parallel networks. We focus on series-parallel networks as this represents bundles of services consisting of complements and substitutes. We show that we can restrict attention to investments in which each provider invests in shortest paths (with respect to the number of services). This means that providers have no incentive to diversify their portfolio and invest in complements, nor to invest in services that can be replaced by less services. The result has some interesting implications. When being interested in either a subgame perfect Nash equilibrium, the social optimum, or combining both in the price of anarchy, it is without loss of generality to assume that the network consists of a single edge. This greatly simplifies the remaining analysis.

We then study the quality of subgame perfect Nash equilibria for different classes of users.
 In Section \ref{sec:fixed_res_value}, we assume that all users have a fixed common willingness to pay. 
We precisely characterize the society's surplus and the providers' profits both in a corresponding optimal solution and in every subgame perfect Nash equilibrium. It turns out that the price of anarchy with respect to society's surplus can be bounded by a constant, whereas the ratio with respect to providers' profits can grow linearly with the number of providers. In Section \ref{sec:elastic_demand}, we redo the same analysis for users that have different willingnesses to pay and obtain similar results as for the fixed reservation value case.

\subsection{Related work}
Our model is closely related to several different classes of recently studied models. There is a close connection with network design games, as introduced by Anshelevich et al.~\cite{anshelevich2003near} and extended by Anshelevich et al.~\cite{anshelevich2008price}. In these games, selfish agents want to connect a personal source to sink by means of a path, and they share the costs of an edge in case multiple agents use that edge in their path. In network investment games, providers do not share the cost of an edge, but rather the revenue from the users using that edge. As a result, we obtain that providers want to invest in paths in equilibrium, an assumption that is not made a-priori. 

\iffalse %local connection games...
Second, we would like to address the connection to network creation games, as introduced by Fabrikant et al.~\cite{fabrikant2003network} and later generalized to, for example, directed networks by Bala and Goyal~\cite{bala2000noncooperative}, and proportional sharing dependent on the number of neighbors by Chauhan et al.~\cite{chauhan2017selfish}. In these games, selfish agents, represented as vertices in a graph, buy edges of a graph and thereby create a network. In network creation games however, connectivity to other agents is the main objective of the players. In network investment games, providers invest in edges so as to attract users in a second stage. 
\todo[inline]{Do we really need this?}
\fi

In network investment games, we assume that users in the second stage behave as in a non-atomic network congestion game. The idea behind congestion games is that selfish users try to minimize their cost in an already existing network. The best-known solution concept for these games is the Wardrop equilibrium~\cite{wardrop1952some}. Beckmann et al.~\cite{beckmann1956studies} and Dafermos and Sparrow~\cite{dafermos1969traffic} proved some structural results of this equilibrium concept. Roughgarden and Tardos~\cite{roughgarden2002bad} were the first to obtain a bound on the price of anarchy for non-atomic routing games.

Recently, starting with Acemoglu and Ozdaglar~\cite{acemoglu2007competition}, several authors considered a two-stage game in which edge owners compete for users by means of tolling. In those games, the users in the second stage also choose a Wardrop equilibrium. Acemoglu and Ozdaglar~\cite{acemoglu2007competition} only considered parallel networks, which was then generalized to parallel-serial networks by Acemoglu and Ozdaglar~\cite{acemoglu2007competitionn}. Ozdaglar~\cite{ozdaglar2008price} and Hayrapetyan et al.~\cite{hayrapetyan2007network} allowed for elastic users. Johari et al.~\cite{johari2010investment} studied an extension that includes entry and investments decisions. Recently, Correa et al.~\cite{correa2018network} and Harks et al.~\cite{HARKS2019} considered the game in which a central authority is allowed to impose price caps. In the former paper, caps are allowed to be different on different edges, in the latter the cap is uniform. Almost all of the above models, with the exception of Acemoglu and Ozdaglar~\cite{acemoglu2007competitionn}, impose the simplifying assumption that the topology of the network is restricted to be parallel. We allow for arbitrary series-parallel graphs and do not impose any restriction on the strategy space of the providers, that is, every provider is allowed to invest in every edge of the network.

A slightly different, but closely related model is the game analyzed by Correa et al.~\cite{correa2014sensitivity}. Here each edge in a series-parallel network represents a producer that competes by selecting a markup that then determines a price function. The main difference to our work is that producers are identified by an edge and thus do not have the possibility to compete on other edges.

Leader-follower models as the previous ones are also known as Stackelberg games \cite{von1934marktform}. Recently, lots of attention has been placed on Stackelberg pricing games. Starting with Labb\'{e} et al.~\cite{labbe1998bilevel}, who considered the problem in which a leader sets prices on a subset of edges and the follower chooses a shortest path, several authors generalized to different combinatorial problems for the follower. For example, Cardinal et al.~\cite{cardinal2011stackelberg} studied the Stackelberg minimum spanning tree problem, and Bil\`{o} el al.~\cite{bilo2008computational} and later Cabello~\cite{cabellostackelberg} investigated the Stackelberg shortest path tree problem. Balcan et al.~\cite{balcan2008item} and Briest et al.~\cite{briest2012stackelberg} considered the power of single-price strategies. B\"{o}hnlein et al.~\cite{bohnlein2017revenue} extended the analysis of this simple algorithm beyond the combinatorial setting.

Lastly, there is a close connection to some of the traditional models of competition in economics, like Cournot competition \cite{cournot1838recherches}. We refer to Harks and Timmermans~\cite{harks2018computing} for some other recent work connecting atomic splittable congestion games to multimarket Cournot oligopolies.
\section{Preliminaries}

We consider a two-sided market in which users purchase bundles of services.
The set of feasible bundles is specified by a technology graph $G_{st}=(V,E)$, which is a directed series-parallel network with source $s \in V$ and sink $t \in V$. Each edge corresponds to a service.
Every $s$-$t$ path in $G_{st}$ describes a feasible bundle of services. Note that serial edges in $G_{st}$ correspond to complements, whereas parallel edges correspond to substitutes.

We model such a market as a two-stage game. In the first stage, a set of providers chooses to invest into services represented by edges. In the second stage, the users select bundles of services or paths in $G_{st}$ given the investments of the providers in the first stage.

The amount of investment affects the speed of the services. We assume for simplicity that prices for each service are charged automatically on a per-time basis. That is, the price function of an edge directly depends on the chosen investments into the corresponding service in the first stage and the current load induced by the users in the second stage. These two quantities highly influence the time that is needed to finish a job on a cloud computing device. Our choice of the price functions and the model of network investment games are motivated by time-dependent prices that are implemented in Amazon Web Services \cite{amazon}, Microsoft Azure \cite{azure} and Google Cloud \cite{google} for using on-demand virtual machines.

 We assume for simplicity that each user is of infinitesimal size which allows us to model the second stage as a non-atomic routing game. However it is important to note that the overall demand is not fixed, but depends on the total price the users have to pay for the services. The demand thus depends on the investment of the providers.

More formally, a \textit{network investment game} is given by a tuple $\mathcal{G}=(G_{st},N,u)$, where $G_{st}=(V,E)$ is a directed series-parallel network with source $s \in V$ and sink $t \in V$, $N=\{1,\ldots,n\}$ is the set of providers, and $u:\mathbb{R}_+\rightarrow\mathbb{R}_+$ is the reservation function. Reservation functions model the users' willingness to pay and impose elastic demand. This concept has already been used in the work of Beckmann, McGuire, and Winsten~\cite{beckmann1956studies}. Intuitively, there is a total demand of $x$ in the game if the total price for bundles of services under demand $x$ is not larger than $u(x)$. We assume that $u$ is non-increasing.

%A directed $s-t$ network is \textit{series-parallel} if it either consists of a single edge $(s,t)$ or can be obtaiend from two series-parallel graphs with terminal $(s_1,t_1)$ and $(s_2,t_2)$ composed either in series or in parallel. In a \textit{series composition}, $t_1$ is identified with $s_2$, $s_1$ becomes $s$ and $t_2$ becomes $t$. In a \textit{parallel composition}, $s_1$ is identified with $s_2$ and becomes $s$, and $t_1$ is identified with $t_2$ and becomes $t$.

%We describe the setting of the game in the following three subsections. First, we define the strategy set of the firms, who we do also call providers. Second, we describe the action of the costumers, called users, that depend on the chosen strategies of the firms. We conclude the description of the game by defining the utility of each firm dependent on the chosen strategies of all firms.

\textbf{The Provider's Strategies.}
Each provider $i \in N$ chooses an investment $b_{i,e}\geq 0$ for every edge $e \in E$. Let $b_i=(b_{i,e})_{e \in E}$ be the investments of provider $i$. Given an investment matrix $b=(b_i)_{i\in N}$, the total investment on edge $e$ is $b_e=\sum_{i \in N} b_{i,e}$. Given $b\in\mathbb{R}_+^{N\times E}$, we define the set of \emph{relevant edges} by $E^+(G_{st}) = \{e \in E(G_{st}) \mid b_e > 0\}$.
The price $c_e$ for a user using an edge $e$, i.e., a service, is an increasing function of the amount of users $f_e$ using edge $e$ and defined by
\[c_e(f_e,b_e) = \begin{cases} \frac{f_e}{b_e} & e\in E^+(G_{st}),\\
\infty & \text{otherwise.}\end{cases}\]
Note that this cost function is a very simple model for prices that depend on the total load of a service and the total investment. Here, we use the simplifying assumption that the time that is needed to complete a service request scales linearly with the load on the machines and inverse linearly with the investment.

\textbf{The User's Actions.}
Let $\Pa(G_{st})$ denote the set of $s$-$t$ paths, and let $\Pa^*(G_{st})$ denote the set of $s$-$t$ paths with minimum number of edges in $G_{st}$. We model the behavior of the users by a flow in $G_{st}$. The flow is defined to be a non-negative vector $f(G_{st})=(f_P)_{P\in\Pa(G_{st})}$. Note that we stick to the path-definition of flows, because we assume strong flow conservation. For a flow $f(G_{st})$ and an edge $e\in E(G_{st})$, let $f_e=\sum_{P \ni e}f_P$ denote the amount of flow on edge $e$. Let $|f(G_{st})|=\sum_{P\in \Pa(G_{st})} f_P$ be the total amount of users that are routed through the network under $f(G_{st})$.

For a given investment matrix $b$, we call some flow $f(G_{st})$ a \textit{Wardrop equilibrium} if for all $P,P' \in \Pa(G_{st})$ with $f_P>0$, we have \[\sum_{e\in P}c_e(f_e,b_e)\leq\sum_{e\in P'}c_e(f_e,b_e)\text{ and }\sum_{e\in P}c_e(f_e,b_e)<\infty.\]

In a network investment game, we assume that both the providers and the users act selfishly. For simplicity, we assume that users are infinitesimal small and thus, a Wardop flow arises in the second stage. We will remark below that, for a given investment matrix and a given non-increasing reservation function, the arising Wardrop flow is uniquely defined, which is necessary for our model to be well-defined. Given the investment matrix $b$, we denote the corresponding Wardrop flow by $f(G_{st},b)$. If $G_{st}$ and/or $b$ is clear from the context, we often omit the dependency on $G_{st}$ and/or $b$ in order to improve readability. If $s$ and $t$ are clear from the context, we sometimes write $G$ instead of $G_{st}$.

\textbf{The Provider's Profits.}
The profit of provider $i \in N$ depends on the associated Wardrop flow $f$, the investment matrix $b$ and the actual cost for the users $c_e(f_e,b_e)$. We define it by 
 \[\pi_i(b)= \sum_{e\in E^+}\frac{b_{i,e}}{b_e}\cdot f_e \cdot c_e(f_e,b_e)- \sum_{e \in E^+} b_{i,e} =  \sum_{e \in E^+} b_{i,e} \cdot\frac{f_e^2}{b_e^2} - \sum_{e\in E^+}b_{i,e}.\]
where $\frac{b_{i,e}}{b_e}$ is the proportional share of provider $i \in N$ and $\sum_{e \in E^+} b_{i,e}$ her investment costs. This profit function is motivated by that fact that each user pays a load-dependent price for using a service of the provider, which is proportionally divided among those providers investing in the service, and the provider has some cost for installing bandwidth for the services. We consider the very simple model that investment costs are linear and the profit is given by the payment of the users minus the investment costs.

We call an investment matrix $b \in \mathbb{R}_+^{N \times E}$ a \textit{subgame perfect Nash equilibrium} if for all $i\in N$ and all $b'_i \in \mathbb{R}_+^E$, 
\[\pi_i(b)\geq\pi_i(b'_i,b_{-i}).\]
Given $b\in\mathbb{R}_+^{N\times E}$, we call $b'_i$ a \textit{better response} for provider $i\in N$ if $\pi_i(b'_i,b_{-i})>\pi_i(b)$. In a subgame perfect Nash equilibrium, no provider has a better response. We call a better response $b'_i$ \textit{demand preserving} if $|f(b'_i,b_{-i})|=|f(b)|$.

\textbf{The Social Welfare}, or social surplus, is traditionally defined as the sum of the providers' surplus and the user's surplus. In our model, similarly as in Hayrapetyan et al.~\cite{hayrapetyan2007network}, this boils down to,
\[ SW(b) = \sum_{i \in N}{\pi_i(b)} + \int_0^{|f|}u(x)\;dx - \sum_{e \in E^+}{f_e c_e(f_e,b_e)} = \int_0^{|f|}u(x)\;dx-\sum_{e\in E}b_e.\]
Note that payments do not appear in the social welfare, as they are transfers from users to providers. An investment matrix $b^*$ is called a \textit{social optimum} if $b^*$ maximizes the social welfare.

For a given instance $\mathcal{G}$, let $\mathcal{E}(\mathcal{G})$ denote the set of subgame perfect Nash equilibria. Notice that $\mathcal{E}(\mathcal{G})$ might be empty. A  natural question is to quantify the loss in social welfare due to competition. Typically, researchers provide upper bounds for the loss in efficiency. We are able to precisely bound this inefficiency for any given instance. To this end, we define the \textit{price of anarchy} and \textit{price of stability} for a given instance as 
\[PoA(\mathcal{G})=\frac{\sup\limits_{b^*\in\mathbb{R}_+^N} SW(b^*)}{\inf\limits_{b\in \mathcal{E}(\mathcal{G})}SW(b)},\text{ and }
PoS(\mathcal{G})=\frac{\sup\limits_{b^*\in\mathbb{R}_+^N} SW(b^*)}{\sup\limits_{b\in \mathcal{E}(\mathcal{G})}SW(b)}.\]
We assume that $0/0=1$ and $PoA(\mathcal{G})=PoS(\mathcal{G})=\infty$ if $\mathcal{E}(\mathcal{G})$ is empty.

A second natural question is to quantify the loss in total profit due to competition. We call this the \textit{providers' price of anarchy} and \textit{providers' price of stability}, and define it by
\[
PPoA(\mathcal{G})=\frac{\sup\limits_{b^*\in\mathbb{R}_+^N} \sum_{i\in N}\pi_i(b^*)}{\inf\limits_{b\in \mathcal{E}(\mathcal{G})}\sum_{i\in N}\pi_i(b)},\text{ and }
PPoS(\mathcal{G})=\frac{\sup\limits_{b^*\in\mathbb{R}_+^N} \sum_{i\in N}\pi_i(b^*)}{\sup\limits_{b\in \mathcal{E}(\mathcal{G})}\sum_{i\in N}\pi_i(b)}.
\]
As before, we assume that $0/0=1$ and $PPoA(\mathcal{G})=PPoS(\mathcal{G})=\infty$ if $\mathcal{E}(\mathcal{G})$ is empty.

We start with some simple observations on Wardrop equilibria for this model. To calculate the flow that models the users, we distinguish between two cases. Either there exists a $s$-$t$ path $P\in \Pa$ such that $b_e > 0$ for all $e \in P$, or not. Put differently, either there is a path $P$ such that for all $e \in P$ we have $e \in E^{+}(G_{st})$, or all paths have investment $0$ for at least one edge. If there is no $s$-$t$ path with strictly positive investment on every edge, the Wardrop equilibrium is the empty flow with $|f| = 0$. If there is some $P \in \Pa$ with $b_e > 0$ on all edges $e \in P$, we obtain the following results. 

\begin{remark}[\cite{beckmann1956studies,dafermos1969traffic}] 
If $f$ is a Wardrop equilibrium, then there is a constant $c_b\geq 0$ such that $\sum_{e\in P}c_e(f_e,b_e)=c_b$ whenever $P\in \Pa$ with $f_e >0$ for all $e\in P$, and $\sum_{e\in P}c_e(f_e,b_e)\geq c_b$ otherwise.
\label{remark:allFlowsSameCost}
\end{remark}

\begin{remark}[\cite{beckmann1956studies,dafermos1969traffic}]  
For two Wardrop flows $f, f'$ with $|f| = |f'|$ in a network with strictly increasing cost functions on every edge, we have that $f_e = f'_e$ for all $e \in E$.
\label{remark:edgeFlowUnique}
\end{remark}

Note that, by Remarks \ref{remark:allFlowsSameCost} and \ref{remark:edgeFlowUnique}, we have that the cost per user in the network investment game $c_b$, as defined in Remark \ref{remark:allFlowsSameCost}, only depends on the total demand $|f|$ for all Wardrop flows $f$ and a fixed $b$. Let us focus on this dependence and write $c_b(|f|)$ instead of $c_b$. It turns out that $c_b(|f|)$ is in fact strictly increasing in $|f|$. The following proposition is due to Lin, Roughgarden, and Tardos~\cite{lin2004stronger}. However, in their variant of the proposition they only assume non-decreasing cost functions and show that $c_b(|f|)$ is non-decreasing in $|f|$. Since we have strictly increasing cost functions, we modify their proof in the appendix such that we get a slightly stronger statement.

\begin{restatable}[\cite{lin2004stronger}]{proposition}{CBStrictlyIncreasing}
$c_b(|f|)$ is strictly increasing in $|f|$.
\label{prop:costStrictlyIncreasing}
\end{restatable}

We note that we can extend Proposition \ref{prop:costStrictlyIncreasing} and show that for network investment games, $c_b(|f|)$ is even linear in $|f|$. The proof uses the series-parallel structure of the graph and can be found in the appendix.

\begin{restatable}{proposition}{fLinearFunction}
Given a vector of investment levels $b$, $c_b(|f|)$ is a linear function of $|f|$.
\label{prop:linearFunction}
\end{restatable}

The above two propositions are very important for the setup of network investment games. We assume that demand directly depends on the actual cost for the users, and thus indirectly on $b$. In order to do so, we use that the costs for each user are given by a linear function $c_b(|f|)$. Combining this with the reservation function $u$ stating that an amount of $u(x)$ has a willingness to pay at least $x$, we define
\[|f|= \inf\{x \in \mathbb{R}_+ \mid u(x) \leq c_b(x)\}.\]
Note that this infimum is well defined, since $u$ is non-increasing and $c_b(\cdot)$ is strictly increasing due to Proposition \ref{prop:costStrictlyIncreasing}. Furthermore, let us illustrate the dependence of $|f|$ on $b$ by the following observation. An increase of some $b_e$ for fixed $f$ induces a change in the corresponding edge cost $c_e(f_e,b_e)$ and weakly decreases the cost of the flow $f$. Since Wardrop flows are optimal flows for linear cost functions, this weakly decreases the value of $c_b(x)$ for any fixed $x$. Overall, one can show that a change of $b$ changes the slope of $c_b(x)$ in the definition above, which then might result in a different demand.

\section{Characterization of Equilibria}\label{sec:char}
In this section, we derive some important properties of subgame perfect Nash equilibria. First, we prove our main result that in every subgame perfect Nash equilibrium, every provider only invests in shortest (with respect to the number of edges) paths. The implications of this result are quite significant and twofold. First, it implies a nice characterization of markets that can be modeled by means of a network investment game. %While assuming existence and convergence to subgame perfect Nash equilibria, it shows that we can somehow assume that the market behaves quite well and structured. 
Second, the result simplifies the further analysis of the game significantly. As long as we are interested in equilibrium investments, the analysis gets notably easier. Instead of considering a strategy space in which every provider is allowed to invest in every edge separately, we can restrict attention to strategies in which providers invest in shortest paths. We will heavily use this characterization in the later sections.

\begin{definition}
We call a strategy  $b_i$ a {\em path strategy} of provider $i$ if there exist values
$ (b_{i,P})_{P \in \Pa}$ with $b_{i,P}\ge 0$ such that $b_{i,e}=\sum_{P: e\in P}{b_{i,P}}$ for all $e \in E$. 
\end{definition}

Note that the definition ensures that we can always decompose a path strategy $(b_{i,e})_{e \in E}$ of a provider $i \in N$ defined on edges into a path decomposition. Thus, when considering path strategies, we can work with the path decomposition of investments, instead of the investments defined on edges. If all providers' investments are path strategies, it is immediately clear that we can also define $b_P=\sum_{i\in N}b_{i,P}$ for all $P\in\Pa$ and $|b|=\sum_{P \in \Pa}{b_P}$.

%Before we proceed to prove the main result of this section, we prove a series of useful lemmas that we will need later in the proof. For all these results it holds that 
Note that, when using the term \emph{shortest paths}, we always refer to shortest paths with respect to the number of edges in the path. Now, we are able to state and prove our main theorem of this section. 
%It states that the only candidate for Nash equilibria are bandwidth vectors that can be written as the sum of bandwidths of shortest paths.

\begin{theorem}\label{thm:sho}
\mbox{}
\begin{enumerate}
\item[(i)] In every subgame perfect Nash equilibrium, every provider $i\in N$ chooses a path strategy on shortest paths.
\item[(ii)] Let $i\in N$. If every provider $j\in N\setminus\{i\}$ chooses a path strategy on shortest paths then the following holds: For every strategy $b_i$ of provider $i$, there is a path strategy on shortest paths $b_i'$ inducing the same amount of flow and at least the same profit for provider $i$.
\end{enumerate}

%If there is a provider $i$ such that $b_i$ can not be written as $b_i = \sum_{P \in \sPa}{b_{i,P}}$, then this provider $i$ has a demand preserving strategy $b^+$ that increases her profit.
\end{theorem}

For the proof of Theorem \ref{thm:sho}, we need the following two lemmas. In the first lemma we observe that in equilibrium a provider cannot choose positive investments on two paths of different length. If so, there is always a better response that only invests in the shorter of the two paths.

\begin{restatable}{lemma}{lemmaParallel}
\label{lem:par}
Let $(G,s,t)$ be a series-parallel graph that is a parallel composition of series-parallel graphs $(G_1,s,t)$ and $(G_2,s,t)$. Let $b$ be a vector of path strategies of paths that all are either shortest $s$-$t$ paths in $G_1$ with length $k$ or shortest $s$-$t$ paths in $G_2$ with length $\ell$. If $k<\ell$, and $b_{i,P_2}>0$ for some provider $i\in N$ and $P_2 \in \Pa^*(G_2)$, then there is a demand preserving better response for $i$ with $b_{i,P_2}=0$.
\end{restatable}

Next we turn to graphs that are a series composition of two subgraphs. Assuming that every provider's strategy can be decomposed into two path strategies for the two subgraphs, we show that in equilibrium they must invest the same value in both subgraphs. Furthermore, if all other providers play shortest path strategies, for all strategies of a provider, the provider always has an at least as good strategy in shortest paths. 
\begin{restatable}{lemma}{lemmaSeries}
\label{lem:ser}
Let $(G,s,t)$ be a series-parallel graph that is a series composition of series-parallel graphs $(G_1,s,v)$ and $(G_2,v,t)$. Let $b$ be a vector of investments $(b_{i,e})_{i \in N, e \in E}$ that can be partitioned into $(b^1_{i,e})_{i \in N, e \in E(G_1)}$ and $(b^2_{i,e})_{i \in N, e \in E(G_2)}$ such that $b^1$ and $b^2$ are shortest path strategies in their corresponding graphs $G_1$ and $G_2$, respectively.
\begin{enumerate}
\item[(i)]  If $\sum_{P \in \Pa^*(G_1)}b^1_{i,P} \neq \sum_{P \in \Pa^*(G_2)}b^2_{i,P}$ for some provider $i \in N$, then there is a demand preserving better response for some $j \in N$. 
\item[(ii)] If, additionally, $\sum_{P \in \Pa^*(G_1)}b^1_{m,P} = \sum_{P \in \Pa^*(G_2)}b^2_{m,P}$  for all providers $m \ne i$, then the following holds. For every strategy $b_i$ there is a shortest path strategy $b_i'$ with $|f(b_{-i}, b_i)| = |f(b_{-i},b_i')|$ and at least the same profit for provider $i$, i.e., $\pi_i(b_{-i}, b_i) \leq \pi_i(b_{-i}, b_i')$.
\end{enumerate}

\end{restatable}

Now, we are ready to prove Theorem \ref{thm:sho}.

\begin{proof}[Proof of Theorem \ref{thm:sho}]
The main difficulty in proving statements on strategy changes of providers lies in the fact that the users' costs and, hence, demand and flow might change. However, we are able to circumvent the dependency on the reservation function by only considering demand preserving better responses. In fact, we prove the first statement of the theorem by showing the following stronger statement. \emph{(i') If there is a provider $i \in N$ such that $b_i$ is not a path strategy on shortest paths, then there is a provider $j \in N$ with a demand preserving better response.}  Note that statement (i') implies part (i) of the theorem.
We show statement (i') and part (ii) of the theorem by induction on the structure of the series-parallel graph $G$. For the base case, a graph with a single edge, both statements are trivially fulfilled. For the induction step, assume the statements are true for series-parallel graphs $G_1$ and $G_2$, and we will argue that the statements are true for a new series-parallel graph that either arises from a parallel composition or a series composition of $G_1$ and $G_2$.

For (i'), observe the following: If there is a provider $i$ such that $b_i$ is not a path strategy in one of the two subgraphs $G_1$ or $G_2$, we can assume w.l.o.g.\ that it is not a path strategy in $G_1$. Then there is a demand preserving better response in $G_1$ for some provider $j$. This is also a demand preserving better response for $j$ in $G$, since it is demand preserving in $G_1$, and thus does not change any flow in $G_2$, and thus any profit in $G_2$. So we can assume that the investments $b_i$ are shortest paths within $G_1$ and within $G_2$ for all $i \in N$.

\textbf{Case 1: Parallel Composition.}
Assume $G$ is the result of a parallel composition of series-parallel networks $G_1$ and $G_2$. % By the induction hypothesis all providers' investments can be written as the sum of investments on shortest paths in $G_1$ and $G_2$, since otherwise there would be a demand-preserving better response in one of the two subgraphs, which would also be a better response for $G$.
%and that there is a provider $i$ such that $b_i$ can not be written as $b_i = \sum_{P \in \Pa}{b_{i,P}}$.
For proving (i'), denote the length of the shortest paths in $G_1$ by $k$ and the length of the shortest paths in $G_2$ by $\ell$. 
If  $k \neq \ell$ and there is a provider $i$ that invests in the longer path, we can apply Lemma \ref{lem:par} to show that there is a demand preserving better response %$b^+_i$
 for provider $i$. For (ii), observe that if $b_j$ is a path strategy on shortest paths for every provider $j \in N \setminus \{i\}$, it is also a path strategy on shortest paths in both subgraphs. Thus, by induction, provider $i$ has a shortest path strategy $b'_i$ that is demand preserving w.r.t. $b_i$ within $G_1$ and $G_2$. Note that this is already a strategy on paths that does not change the overall demand and yields at least as much profit as $b_i$. Note that this is not a shortest path strategy in $G$ only if $k \neq \ell$, say w.l.o.g.\ $k < \ell$ and $i$ invests in $G_2$. In this case, we apply Lemma \ref{lem:par} and conclude that $i$ has a demand preserving better response and, thus in fact has the desired strategy on shortest paths in $G$.

\textbf{Case 2: Series Composition.}
Assume $G$ is the result of a series composition of series-parallel networks $G_1$ and $G_2$. For (i'), note that have already argued that $b$ is a shortest path strategy within $G_1$ and within $G_2$ for every provider. However, if $b_i$ it is not a shortest path strategy in $G$ for some $i\in N$, we fulfill the conditions of Lemma \ref{lem:ser} (i).

For (ii), note that we also fulfill the additional condition of Lemma \ref{lem:ser} (ii). By using the lemma, we can conclude that for every strategy $b_i$, provider $i$ has a strategy $b_i'$ on shortest paths inducing the same demand and at least the same profit. 
\end{proof}

Note that we have proven that in a subgame perfect Nash equilibrium, all providers invest in shortest path strategies. Additionally, for proving that some investment matrix is an equilibrium, we can restrict the strategy set to shortest path strategies.

 We proceed with a proposition that significantly reduces the set of shortest path strategies. We show that if each provider chooses an investment on shortest paths, it is irrelevant how this investment is distributed over the shortest paths. The proof of the proposition is moved to the appendix.

\begin{restatable}{proposition}{investmentsDoNotDependOnPath}
\label{prop:equivalence}
Let $b$ and $b'$ be investment matrices such that for all $i\in N$, $b_i$ and $b'_i$ are path strategies, and $b_i$ and $b'_i$ can be decomposed into path decompositions $(b_{i,P})_{P \in \Pa^*}$ and $(b'_{i,P})_{P \in \Pa^*}$ only using shortest paths, respectively. Additionally, let $\sum_{P \in \Pa^*}{b_{i,P}} = \sum_{P \in \Pa^*}{b'_{i,P}}$ for all $i \in N$. Then, $\pi_i(b')=\pi_i(b)$ for all $i \in N$ and $|f(b)| = |f(b')|$. 
\end{restatable}

By using Proposition \ref{prop:equivalence}, we can classify shortest path strategies into equivalence classes, represented by a single number. 

\begin{observation}
\label{obs:singleNumber}
In order to analyze subgame perfect Nash equilibria and check their existence, it is sufficient to check investments on shortest paths. By Proposition \ref{prop:equivalence}, we know that it is irrelevant on which shortest path a provider invests. Thus, we can assume that each strategy of a provider is not a vector of investments $(b_{i,e})_{e\in E}$, but a single number $\sum_{P \in \Pa^*}{b_{i,P}}$ representing the total investment in shortest paths.
\end{observation}

\section{Homogeneous Users}\label{sec:fixed_res_value}
In this section, we assume that all users are homogeneous and have the same fixed reservation value $R \in \mathbb{R}_+$ for processing a task, and decide not to process their task if the price is above $R$. More formally, we assume that $u(x)=R$ for all $x\in[0,d]$, and $u(x)=0$ for all $x>d$, where $d\in\mathbb{R}_+$ represents the size of the population of users. We show that there always exists a subgame perfect Nash equilibrium and analyze the inefficiency of equilibria in terms of the PoA and PPoA.

Recall that, by Observation \ref{obs:singleNumber}, we can restrict the strategy spaces of the providers to path strategies on shortest paths and denote them by a single number. We slightly abuse notation and denote the strategy chosen by provider $i$ by $b_i \in \mathbb{R}_+$. 

We first compute the demand and profits of providers in a network investment game with a fixed reservation value.
\begin{restatable}{lemma}{lemmaprofix}\label{lem:prof}
For $u(x)=R$ for all $x\in[0,d]$, and $u(x)=0$ for all $x>d$, where $d\in\mathbb{R}_+$,
\begin{itemize}
 \item[(i)]  the demand $|f(b)|$ is given by the function
$|f(b)|=\min\left\{\frac{R\cdot |b|}{k},d\right\}$ and
\item[(ii)] the profit $\pi_i$ is given by $
\pi_i(b_i, b_{-i})=\min\left\{k\cdot\left(\frac{R^2}{k^2}-1\right)\cdot b_i,k\cdot\left(\frac{d^2}{|b|^2}-1\right)\cdot b_i\right\}$,
 \end{itemize}
where $k$ denotes the length of a shortest $s$-$t$ path.
\end{restatable}

Using the closed form for the demand and the profit, we show by a case distinction of $R<k$, $R = k$ and $R > k$ in the appendix that a subgame perfect Nash equilibrium always exists and has the same total investment.
\begin{restatable}{theorem}{FixedExistenceOfNash}\label{thm:exi}
There always exists a subgame perfect Nash equilibrium.
\end{restatable}

The following example illustrates that a subgame perfect Nash equilibrium need not be unique.
\begin{example}
Let $G$ consist of a single edge connecting $s$ to $t$, $n=2$, and $u(x)=2$ for all $x\in[0,1]$, and $u(x)=0$ for all $x>1$. Then for $i=1,2$, $\pi_i(b)=\min\left\{3b_i,\left(\frac{1}{|b|^2}-1\right)\cdot b_i\right\}$. From the first order conditions of provider $i=1,2$, we get that all $b\in\mathbb{R}_+^2$ with $\frac{3}{16}\leq b_1\leq\frac{5}{16}$ and $b_1+b_2=\frac{1}{2}$ are subgame perfect Nash equilibria. Notice that the total investment is the same in all subgame perfect Nash equilibria.
\end{example}

Given that a subgame perfect Nash equilibrium exists, we might wonder about the induced performance of an equilibrium. In order to answer this question, we consider two different measures of performance: the social welfare and the total providers' profits.

The main result of this section gives a tight characterization on the price of anarchy for fixed reservation value users.
\begin{theorem}
For an instance $\mathcal{G}$ with $u(x)=R$ for all $x\in[0,d]$, and $u(x)=0$ for all $x>d$, where $d\in\mathbb{R}_+$,
\begin{align*}
PoA(\mathcal{G})=PoS(\mathcal{G})=\begin{cases}
\frac{(R+k)\cdot(R-k)}{R\cdot(R-k\sqrt{\frac{n-2}{n}})}&\text{ if }R>k\text{ and }\sqrt{\frac{n-2}{n}}\geq \frac{k}{R},\\
1&\text{ otherwise.}
\end{cases}
\end{align*}
\end{theorem}
\begin{proof}
First, observe from the definition of the social welfare that $SW(b)=b\cdot\left(\frac{R^2}{k}-k\right)$ for all $b\in\mathbb{R}_+^N$ with $|b|\leq\frac{d\cdot k}{R}$. Moreover, $SW(b)$ is decreasing in $|b|$ for $|b|>\frac{d\cdot k}{R}$. This implies that $\sup\limits_{b \in \mathbb{R}_+^{N \times E}} SW(b)$ is attained at some $b^*\in\mathbb{R}_+^N$ with $|b|^*\leq\frac{d\cdot k}{R}$. In particular, if $R\leq k$, all subgame perfect Nash equilibrium maximize the social welfare.

So assume that $R>k$. Then, $|b|^*=\frac{d\cdot k}{R}$. From the proof of Theorem \ref{thm:exi}, it follows that either $|b|=\frac{d\cdot k}{R}$ or $|b|=d\cdot\sqrt{\frac{n-2}{n}}$. In the former case, we have that all subgame perfect Nash equilibrium maximize the social welfare. In the latter case, we have $\frac{d\cdot R-k\cdot \frac{d\cdot k}{R}}{d\cdot R-k\cdot d\cdot \sqrt{\frac{n-2}{2}}}=\frac{(R+k)\cdot(R-k)}{R\cdot(R-k\sqrt{\frac{n-2}{n}})}\leq1+\frac{k}{R}<2$, where the last inequality follows since $R> k$. 
\end{proof}
\begin{corollary}
For all $\mathcal{G}$ with $u(x)=R$ for all $x\in[0,d]$, and $u(x)=0$ for all $x>d$, where $d\in\mathbb{R}_+$, $PoA(\mathcal{G})<2$.
\end{corollary}
Note that the above result is quite robust, and states that, independent of $R$, the social welfare under competition is at least $1/2$ the optimal social welfare.

However, the following result shows that the providers' price of anarchy is not bounded by a constant, but could grow linearly with the number of providers (by at most $n/2$), even for fixed reservation value users.
\begin{theorem}
For an instance $\mathcal{G}$ with $u(x)=R$ for all $x\in[0,d]$, and $u(x)=0$ for all $x>d$, where $d\in\mathbb{R}_+$,
\begin{align*}
PPoA(\mathcal{G})=PPoS(\mathcal{G})=\begin{cases}
\frac{n\cdot\sqrt{\frac{n-2}{n}}\cdot(R^2-k^2)}{2\cdot R\cdot k}&\text{ if }R>k\text{ and }\sqrt{\frac{n-2}{n}}\geq \frac{k}{R},\\
1&\text{ otherwise.}
\end{cases}
\end{align*}
\end{theorem}
\begin{proof}
First, observe that if $R\leq k$, then $\sum_{i\in N}\pi_i(b)\leq 0$ for all $b\in\mathbb{R}_+^N$ and thus all subgame perfect Nash equilbria maximize the sum of providers' profits.

Second, if $R>k$, then by Lemma \ref{lem:prof},  $b^*=\argmax\limits_{b\in\mathbb{R}_+^N} \sum_{i\in N}\pi_i(b)$ has $|b^*|=\frac{d\cdot k}{R}$. In particular, if $\sqrt{\frac{n-2}{n}}< \frac{k}{R}$, all subgame perfect Nash equilibrium maximize the sum of the providers' profits.

Third, if $R>k$ and $\sqrt{\frac{n-2}{n}}\geq\frac{k}{R}$, then from the proof of Theorem \ref{thm:exi}, it follows that in a subgame perfect Nash equilibrium, $|b|=d\cdot\sqrt{\frac{n-2}{n}}$. Hence, we get

$\frac{k\cdot\left(\frac{d^2}{\left(\frac{d\cdot k}{R}\right)^2}-1\right)\cdot \frac{d\cdot k}{R}}{k\cdot\left(\frac{d^2}{\left(d\cdot\sqrt{\frac{n-2}{n}}\right)^2}-1\right)\cdot d\cdot\sqrt{\frac{n-2}{n}}}=\frac{n\cdot\sqrt{\frac{n-2}{n}}\cdot(R^2-k^2)}{2\cdot R\cdot k}$.
\end{proof}

\section{Heterogeneous Users}\label{sec:elastic_demand}

We now turn to reservation functions that correspond to heterogeneous users with differing reservation values. As an example we choose one of the arguably simplest class of functions. We note that the study of different and application specific functions is an interesting open question. Here, we consider the reservation functions $u(x)=\frac{1}{x^{1/\alpha}}$, where $\alpha>0$, for all $x\in\mathbb{R}_+$. We distinguish between one and multiple provider games as for one provider games a subgame perfect Nash equilibrium might not exist. We study equilibrium existence, equilibrium uniqueness and the inefficiency in terms of the PoA and PPoA. Note that we again restrict the strategy spaces of the providers to path strategies on shortest paths and denote them by a single number $b_i \in \mathbb{R}_+$ due to Observation \ref{obs:singleNumber}. 

We show that for network investment games with a single provider a monopoly equilibrium exists if and only if $\alpha>1$. If $\alpha \leq 1$, it is always beneficial for the provider to decrease investment to a smaller, strictly positive amount. For the case of $n\geq 2$ there is a unique subgame perfect equilibrium. 
\begin{restatable}{theorem}{theoremOne}\label{thm:one}
Let $k$ denote the length of a shortest $s$-$t$ path in $G$. For $n=1$, there exists a subgame perfect Nash equilibrium if and only if $\alpha>1$. If $\alpha >1$, the equilibrium is unique and given by
%\[
$b_1=\frac{\left(1-\frac{2}{(\alpha+1)}\right)^{\frac{\alpha+1}{2}}}{k^\alpha}$.%\]

For $n\geq 2$, there is a unique Nash equilibrium given by $b_i=\frac{1}{|N|}\cdot\frac{\left(1-\frac{2}{(\alpha+1)\cdot |N|}\right)^{\frac{\alpha+1}{2}}}{k^\alpha}$ for all $i\in N$. 
\end{restatable}

For $n=1$, we refer to the appendix. For the proof of the theorem for $n\geq 2$, we need the following lemmas that are proven in the appendix.

\begin{restatable}{lemma}{propositionInv}\label{lem:inv}
Let $S$ be the set of providers $i\in N$ with $b_i>0$ for some investment vector $b$. Then, if there is some $i\in S$ with
$b_i \neq \frac{1}{|S|}\cdot\frac{\left(1-\frac{2}{(\alpha+1)\cdot |S|}\right)^{\frac{\alpha+1}{2}}}{k^\alpha}$,
then $b$ is not a subgame perfect Nash equilibrium.
\end{restatable}

\begin{restatable}{lemma}{propositionAllinvest}\label{lem:all}
Let $S$ be the set of providers $i\in N$ with $b_i>0$. If $S\neq N$, then this is not a subgame perfect Nash equilibrium.
\label{lem:allInvest}
\end{restatable}

\begin{proof}[Proof of Theorem \ref{thm:one}]
Assume $n\geq 2$. First, note that by Lemma \ref{lem:allInvest} all providers have to invest and by Lemma \ref{lem:inv} in a Nash equilibrium all have to invest $\tilde{b}=\frac{1}{|N|}\cdot\frac{\left(1-\frac{2}{(\alpha+1)\cdot |N|}\right)^{\frac{\alpha+1}{2}}}{k^\alpha}$. Thus, $b= \tilde{b} \cdot \mathbbm{1} \in \mathbb{R}^N$ is the only possibility for a Nash equilibrium. It remains to show that $b$ is in fact a Nash equlibrium. In order to do so, we fix some provider $i \in N$ and prove that $\tilde{b}$ is in fact a best response to $b=\tilde{b} \cdot \mathbbm{1}$. First, we observe in  Lemma \ref{lem:con} in the appendix that $\pi_i(b_i,b_{-i})$ is continuous in $b_i$. We proceed by showing that (i) $\tilde{b}$ is the only value fulfilling the first order conditions and (ii) yields a positive profit. For (i) note that
 \begin{align*}
 \frac{\partial \pi_i(b_i,b_{-i})}{\partial b_i} &= \frac{k^{\frac{1-\alpha}{\alpha+1}}}{(b_i+|b|_{-i})^{\frac{2}{\alpha+1}}}\cdot\left(1-\frac{2}{(\alpha+1)}\cdot\frac{b_i}{b_i+|b|_{-i}}\right)-k.
 \end{align*}
 We have shown in Lemma \ref{lem:inv} that $\tilde{b}$ fulfills the first order condition for all $i \in N$. For showing uniqueness, define $h(x)=\left(1-\frac{2}{(\alpha+1)}\cdot\frac{x}{x+|b|_{-i}}\right)$ and observe that if $h(\hat{x}) < 0$ for some $\hat{x}$, then $h(x) < 0$ for all $x \geq \hat{x}$. We conclude that all values $b^*$ fulfilling the first order condition have the property $b^* < \hat{x}$. However, for all $b^* < \hat{x}$, the $h(b^*)$ is decreasing in $b^*$. This shows that there is in fact only one value fulfilling the first order condition.
 
 For (ii), observe that $\pi_i(b)$ evaluates to $\tilde{b}\left(k\left(1-\frac{2}{(\alpha+1)n}\right)^{-1}-k\right)$, which is positive since $\left(1-\frac{2}{(\alpha+1)n}\right)^{-1}>1$.
\end{proof}

We are now ready to precisely quantify the inefficiency of subgame perfect Nash equilibria for any instance of a network investment game with reservation function $u(x)=\frac{1}{x^{1/\alpha}}$, where $\alpha>1$. The proof can be found in the appendix. As in immediate consequence, we obtain that the price of anarchy is upper bounded by  a surprisingly small constant of approximately $1.22$ for every graph and every $\alpha>1$.
\begin{restatable}{theorem}{PoAElastic}
For an instance $\mathcal{G}$ with  $u(x)=\frac{1}{x^{1/\alpha}}$, where $\alpha>1$,
\[PoA(\mathcal{G})=PoS(\mathcal{G})=\left(\frac{\alpha n}{n(\alpha+1)-2}\right)^{\frac{\alpha-1}{2}}\frac{2\alpha n}{n(\alpha+1)+2(\alpha-1)}\leq 2\sqrt{\frac{1}{e}}.\]
\end{restatable}

As an interesting consequence, we would like to point out that the equilibrium for a two providers game is in fact optimal.
\begin{corollary}
For $u(x)=\frac{1}{x^{1/\alpha}}$, where $\alpha>1$, and $n=2$, we have $PoA=1$.
\end{corollary}

We now consider the total providers' profit and derive tight bounds on the PPoA. Interestingly, in contrast to the PoA, the PPoA is not  constant but grows almost linearly with the number of providers.

\begin{restatable}{theorem}{theoremPPOA}\label{thm:ppoa2}
For an instance $\mathcal{G}$ with $u(x)=\frac{1}{x^{1/\alpha}}$, where $\alpha>0$, and $n\geq 2$,
\[PPoA(\mathcal{G})=PPoS(\mathcal{G})=\begin{cases}
\infty&\text{ if }\alpha<1,\\
n\left(\frac{(\alpha+1)n-2}{n(\alpha-1)}\right)^{\frac{1-\alpha}{2}} &\text{otherwise.}
\end{cases}\]
\end{restatable}

\section{Conclusion}
We have considered a class of games, called network investment games, in which providers invest in bandwidth of particular services that are then used by users to process a specific task. We studied the existence and inefficiency of subgame perfect Nash equilibria when the underlying technology graph is series-parallel, reflecting that services can be either complements or substitutes. We showed that it is essentially without loss of generality to restrict attention to one-edge networks. For two particular classes of reservation functions we studied the performance of subgame perfect Nash equilibria and we quantified this performance in terms of the social welfare and the total providers' profit.

Several questions remain open. We have given examples for which subgame perfect Nash equilibria exist and do not exit, and we have shown that for series-parallel graphs the existence of subgame perfect Nash equilibria is essentially independent of the technology graph. It would be interesting to know if there is a (general) property of the reservation functions that characterize the existence of subgame perfect Nash equilibria. In case there is such an existence condition, it would be nice to find a price of anarchy bound that is independent of the class of reservation functions being used? Unfortunately, our results have far less implications for graphs that are not series-parallel, like for example the Braess' graph. Is it still true that for more general networks providers only invest in shortest paths in equilibrium?

%%
%% Bibliography
%%

%% Please use bibtex, 

\bibliography{lipics-v2019}

\appendix
\section{Omitted Proofs and Lemmas}
\section*{Preliminaries}

\CBStrictlyIncreasing*
\begin{proof}
Given an investment matrix $b$, let $f,f'$ be corresponding Wardrop flows with $|f|<|f'|$. We will prove that there exists a $P\in\Pa$ such that $f_e<f'_e$ for all $e\in P$. Suppose not, then we construct a set of vertices $S$ with $s\in S$ and $t\notin S$ as follows. First, add $s$ to $S$. As long as there is an edge $e$ from $S$ to $u \in V\setminus S$ with $f_e<f'_e$ add $u$ to $S$. Since there is at least one edge $e \in P$ with $f_e\geq f'_e$ for all $P \in \Pa$, we have that $t$ will never be added to $S$. Thus, $S$ is a proper subset of $V$ and for all edges $e$ from $S$ to $V\setminus S$, $f_e\geq f'_e$. Given that $G$ is series-parallel, there are no edges from $V\setminus S$ to $S$, which would imply that $|f|\geq|f'|$ which is a contradiction.

Since there exists a $P\in\Pa$ such that $f_e<f'_e$ for all $e\in P$ and cost functions are strictly increasing by definition, the proposition follows by Remark \ref{remark:allFlowsSameCost}.
\end{proof}

\fLinearFunction*

Before we can prove the proposition by means of induction, we need the following structural lemma about series-parallel graphs.

\begin{lemma}\label{lem:stillSePaGraph}
Let $G_{st}$ be a series-parallel graph and $\Pa$ be the set of simple $s$-$t$ paths in $G_{st}$. Let $G'_{st}$ arise from $G_{st}$ by choosing any non-empty set $\Pa' \subseteq \Pa$ and defining $G'=(V(G), \{e \in E(G) \mid e \text{ is in at least one path in }\Pa' \})$. Then $G'_{st}$ is a series-parallel graph. 
\end{lemma} 
\begin{proof}
We show the lemma by induction on the series-parallel graph structure. First, note that the lemma is clearly fulfilled for a series-parallel graph with a single edge. Next, assume that $G_{st}$ arises from a series composition of series-parallel graphs $(G_{sv}$ and $(G_{vt})$. In this case, each path in $G_{st}$ also forms a path in $G_{sv}$ and $G_{vt}$, so $G'_{sv}$ and $(G'_{vt}$ are series-parallel, where $G'_{sv}=(V(G_{sv}), \{e \in E(G_{sv}) \mid e \text{ is in at least one path in }\Pa' \})$ and $G'_{vt}=(V(G_{vt}), \{e \in E(G_{vt}) \mid e \text{ is in at least one path in }\Pa' \})$. Due to the definition of $G'_{sv}$ and $G'_{vt}$, the graph $G'_{st}$ can be constructed by a series composition of $G'_{sv}$ and $G'_{vt}$, thus is also series-parallel.

Now, assume that $G_{st}$ arises from a parallel composition of series-parallel graphs $(G^1_{st}$ and $G^2_{st}$. In this case, each path in $G_{st}$ forms an $s$-$t$ path in either $G^1_{st}$ or $G^2_{st}$, so ${G^1}'_{st}$ and ${G^2}'_{st}$ are series-parallel, where ${G^1}'$ and ${G^2}'$ are defined as above. Due to the definition of ${G^1}'$ and ${G^2}'$, the graph $(G'_{st}$ can be constructed by a parallel composition of ${G^1}'_{st}$ and ${G^2}'_{st}$, thus is also series-parallel.  
\end{proof}

\begin{proof}[Proof of Proposition \ref{prop:linearFunction}]
Let $G_{st}$ be a series-parallel graph and $b$ some investment vector. First, we do the following preprocessing. Observe that the flow particles can only use edges $e \in E$ that are on at least one $s$-$t$ path $P$ with $b_e>0$ for all $e \in P$. Thus, we can ignore and delete all other edges without changing $c_b(|f|)$. By using Lemma \ref{lem:stillSePaGraph}, we observe that the remaining graph is series-parallel. In the following, we show the Proposition by induction on the series-parallel structure.

First, note that if $E(G)=\{(s,t)\}$, i.e., $E(G)$ contains only the edge $e=(s,t)$, then the cost function $c_b$ is given by definition of the cost function of edge $e$, since $|f| = f_e$. We have
\[c_b(|f|) = c_e(|f|,b_e) = \frac{|f|}{b_e},\]
which is a linear function in $|f|$. Now suppose $G_{st}$ arises from a series or a parallel composition of series-parallel graphs $G^1_{sv}$ and $G^2_{vt}$, or $G^1_{st}$ and $G^2_{st}$, respectively. We assume that the Proposition holds for the two subgraphs and we have $c_b(|f|) = g_1 \cdot |f|$ for $G^1$ and $c_b(|f|) = g_2 \cdot |f|$ for $G^2$, where $g_1$ and $g_2$ are the corresponding slopes that do not depend on $f$. In the following, we calculate the cost function $c_b(|f|)$ in $G$. First, consider the case of a series composition. Then, the total cost of any flow particle is clearly $c_b(|f|) = g_1 \cdot |f| + g_2 \cdot |f|$, which is a linear function in $|f|$ with slope $g_1+g_2$. Second, for a parallel composition, note that a particle either uses graph $(G^1_{st}$ or $(G^2_{st}$. Both routes have the same cost, due to the definition of the Wardrop flow. Assume w.l.o.g.\ that the particle travels through $(G^1_{st}$. Let $|f_1|$ denote the amount of particles using $G^1$ and $|f_2|$ the amount of particles using $G^2$, respectively. Then, $c_b(|f|) = g_1 \cdot |f_1|$, where $|f_1|$ arises such that $g_1 |f_1| = g_2 |f_2|$ with $|f_1| + |f_2| = |f|$. We conclude $|f_1| = \frac{g_2}{g_1+g_2}|f|$, i.e.,\ \[c_b(|f|) = \frac{g_1 g_2}{g_1+g_2}\cdot|f|,\]
which is a linear function in $f$ and finishes the proof.  
\end{proof}

\section*{Characterization of Equilibria}
\investmentsDoNotDependOnPath*
\begin{proof}
We calculate the profit of some provider $i \in N$ and show that it only depends on the sum of the invested bandwidths. Using the Wardflow $f$ as defined in Lemma \ref{lem:flo}, we obtain
\begin{align*}
\pi_i(b)=  \sum_{e \in E^+} b_{i,e} \cdot\frac{f_e^2}{b_e^2} - \sum_{e\in E^+}b_{i,e}=\sum_{e \in E^+} b_{i,e}\cdot\left(\left(\frac{|f|}{|b|}\right)^2-1\right).
\end{align*} 
Note that $\pi_i(b)$ only depends on $|f|$, $\sum_{e \in E^+}b_{i,e}$ and $|b|$. We can thus conclude that the actual choice of shortest paths is irrelevant as long as the total investment level of a provider remains the same.

We show in the following Lemma \ref{lem:flo}, that the cost of each path only depends on $k$, $|f(b)|$ and $|b|$. Given that $|b|=|b'|$, we can conclude that $|f(b)|=|f(b')|$.
\end{proof}

\begin{lemma}\label{lem:flo}
Given that all providers choose path strategies on shortest paths, the flow $(f_P)_{P \in \Pa}$ with
\[f_P = \frac{b_P}{|b|} \cdot |f(b)|\]
is a Wardrop flow. The corresponding price of edge $e$ for a user equals $c_e(f_e,b_e)=\frac{|f|}{|b|}$ for all $e \in E^+$.
\end{lemma}

\begin{proof}
Let $k$ be the length of all shortest paths. Since $G$ is series-parallel, the graph $(V,E^+)$ is a graph with the property that all $s$-$t$ paths in this graph have length $k$. For all other $P\in\mathcal{P}$, the cost of the path equals $\infty$ by definition, since there must be an edge $e \in P$ that is not part of any shortest path, and thus without bandwidth. It remains to consider the paths $P \in \Pa^*$. First, let $P \in \Pa^*$ with $e \in E^+$ for all $e \in P$. We observe
\begin{align*}
c_e(f_e,b_e) &= \frac{f_e}{b_e} = \frac{\sum_{\tilde{P}: e \in \tilde{P}}{f_{\tilde{P}}}}{\sum_{\tilde{P} : e \in \tilde{P}}{b_{\tilde{P}}}}= {\frac{\sum_{\tilde{P}: e \in \tilde{P}}{\frac{b_{\tilde{P}}}{|b|} \cdot |f|}}{\sum_{\tilde{P}: e \in \tilde{P}}{b_{\tilde{P}}}}}=\frac{|f|}{|b|},
 \end{align*}
 which proves the second part of the theorem. We conclude that
 \begin{align*}
 \sum_{P \in \Pa}{c_e(f_e,b_e)} =  \sum_{P \in \Pa}{\frac{|f|}{|b|}} = k \frac{|f|}{|b|}.
 \end{align*}
Thus the cost of any path $P$ with $e \in E^+$ for all $e \in P$ only depends on $k$, $|f|$ and $|b|$.  Note that all paths $P$ with some $e \in (E \setminus E^+) \cap P$ have cost equal to $\infty$ by definition. We conclude that there is no path in $G$ with smaller cost than the used paths in the flow constructed above. Assuming that all providers choose path strategies on shortest paths, the constructed flow $f$ is a Wardrop equilibrium.
\end{proof}

\lemmaParallel*

\iffalse
\begin{lemma}\label{lem:par}
Let $P_1\in\Pa$ with $|e \in P_1| = k$ and $P_2\in\Pa$ with $|e \in P_2| = \ell$ and $P_1 \cap P_2 = \emptyset$. Let $b$ be a vector of path strategies. If $k<\ell$, and $b_{i,P_2}>0$ for some provider $i \in N$, then there is a demand preserving better response for $i$.
\end{lemma}
\fi

\begin{proof}
Define $|b|_{G_1}=\sum_{P\in \Pa(G_1)}b_P$ and $|b|_{G_2}=\sum_{P\in \Pa(G_2)}b_P$. We will prove that the flow $f$ defined by
\[f_P=\frac{\frac{b_P}{k}}{\frac{|b|_{G_1}}{k} + \frac{|b|_{G_2}}{\ell}} \cdot |f|\]
for all $P\in \Pa(G_1)$, and 
\[f_P=\frac{\frac{b_P}{\ell}}{\frac{|b|_{G_1}}{k} + \frac{|b|_{G_2}}{\ell}} \cdot |f|\]
for all $P\in \Pa(G_2)$, is a Wardrop flow.

We will show the proof only for paths $P\in\Pa(G_1)$. For all $P\in\Pa(G_2)$, the proof is analogous. 
Notice that for all $P\in \Pa(G_1)\setminus \Pa^*(G_1)$, the costs are infinity as there is at least one edge without investment. For all $P\in\Pa^*(G_1)$, we have
\begin{align*}
\sum_{e \in P}\frac{f_e}{b_e} &=\sum_{e \in P}\frac{\sum_{\tilde{P}: e \in \tilde{P}}{f_{\tilde{P}}}}{\sum_{\tilde{P} : e \in \tilde{P}}{b_{\tilde{P}}}}\\
&= \sum_{e \in P}{\frac{\sum_{\tilde{P}: e \in \tilde{P}}{\frac{\frac{b_{\tilde{P}}}{k}}{\frac{|b|_{G_1}}{k} + \frac{|b|_{G_2}}{\ell}} \cdot |f|}}{\sum_{\tilde{P}: e \in \tilde{P}}{b_{\tilde{P}}}}}\\
&= \sum_{e \in P}{\frac{\frac{|f|}{k}}{\frac{|b|_{G_1}}{k} + \frac{|b|_{G_2}}{\ell}}}\\
&=\frac{|f|}{\frac{|b|_{G_1}}{k} + \frac{|b|_{G_2}}{\ell}}.
 \end{align*}
Thus the cost of all used paths are the same, and hence the constructed flow is a Wardrop flow.

Now, suppose that there is a provider $i\in N$ with $b_{i,P_2}>0$ for some $P_2\in\Pa(G_2)$. We show that provider $i$ has a demand preserving better response. First, note that the total profit of provider $i$ equals
\begin{align*}
\pi_i(b)&=  \sum_{e \in E^+} b_{i,e} \cdot\frac{f_e^2}{b_e^2} - \sum_{e\in E^+}b_{i,e}\\
&=  \sum_{P \in \Pa : b_{i,P}>0}\sum_{e \in P} b_{i,P} \cdot\frac{f_e^2}{b_e^2} - \sum_{e\in E^+}b_{i,e}\\
&=  \sum_{P \in \Pa : b_{i,P}>0}b_{i,P} \sum_{e \in P} \left(\frac{\frac{|f|}{|P|}}{\frac{|b|_{G_1}}{k} + \frac{|b|_{G_2}}{\ell}}\right)^2 - \sum_{e\in E^+}b_{i,e}\\
&=  \sum_{P \in \Pa : b_{i,P}>0}b_{i,P} \frac{1}{|P|}\left(\frac{|f|}{\frac{|b|_{G_1}}{k} + \frac{|b|_{G_2}}{\ell}}\right)^2 - \sum_{e\in E^+}b_{i,e}\\
&=  \left(\frac{|f|}{\frac{|b|_{G_1}}{k} + \frac{|b|_{G_2}}{\ell}}\right)^2 \sum_{P \in \Pa : b_{i,P}>0} \frac{b_{i,P}}{|P|} - \sum_{e\in E^+}b_{i,e}
\end{align*}
Next, we show that $i$ has a demand preserving better response by shifting some investment from $P_2$ to a path $P_1 \in \Pa^*(G_1)$. Let $b'_{i,P_1}=b_{i,P_1}+\frac{k}{\ell}\cdot b_{i,P_2}$ on path $P_1$ and $b'_{i,P_2}=0$ on $P_2$. We observe that $|f(b)| = |f(b')|$, since
\begin{align*}
c_b(|f|) = \frac{|f|}{\frac{|b|_{G_1}}{k} + \frac{|b|_{G_2}}{l}} = \frac{|f|}{\frac{|b|_{G_1}+\frac{k}{\ell}b_{i,P_2}}{k} + \frac{|b|_{G_2}-b_{i,P_2}}{l}} = c_{b'}(|f|)\;.
\end{align*}
The total difference in the profit of provider $i$ is
\begin{align*}
& \pi_i(b'_i,b_{-i}) - \pi_i(b_i,b_{-i})\\
&=   \left(\frac{|f|}{\frac{|b|_{G_1}}{k} + \frac{|b|_{G_2}}{\ell}}\right)^2 \left(\frac{k}{\ell}\frac{b_{i,P_2}}{k} - \frac{b_{i,P_2}}{\ell} \right) -  \left(k \frac{k}{\ell}b_{i,P_2}-\ell b_{i,P_2}\right)\\
&= b_{i,P_2} \left(\ell-\frac{k^2}{\ell}\right).
\end{align*}
Since $k<\ell$, this term is positive and the deviation is profitable.
\end{proof}

\lemmaSeries*
\begin{proof}
We start the proof by showing (i).
Define $|b_i|_{G_1}=\sum_{P \in \Pa^*(G_1)}b^1_{i,P}$ and $|b_i|_{G_2}=\sum_{P \in \Pa^*(G_2)}b^2_{i,P}$ for all $i\in N$, and $|b|_{G_1}=\sum_{i\in N}|b_i|_{G_1}$ and $|b|_{G_2}=\sum_{i\in N}|b_i|_{G_2}$. Observe that in a subgame perfect Nash equilibrium if $|f|=0$, then either $|b|_{G_1}=0$ or $|b|_{G_2}=0$. In that case, provider $i\in N$ with $|b_i|_{G_1}\neq |b_i|_{G_2}$ has a demand preserving better response by not investing at all. So we can assume that $|f|>0$. W.l.o.g.\ suppose that $|b|_{G_1}\leq |b|_{G_2}$. We show that there exists a provider $j\in N$ that has a profitable deviation.

Since there is a provider $i\in N$ with $|b_i|_{G_1}\neq |b_i|_{G_2}$, there is a provider $j\in N$ with $|b_j|_{G_1}<|b_j|_{G_2}$ and $\frac{|b_j|_{G_1}}{|b|_{G_1}}\leq \frac{|b_j|_{G_2}}{|b|_{G_2}}$. This provider exists, because $\sum_{i\in N}\frac{|b_i|_{G_1}}{|b|_{G_1}}=\sum_{i\in N}\frac{|b_i|_{G_2}}{|b|_{G_2}}=1$, and for all providers $m\in N$ with $|b_m|_{G_1}\geq |b_m|_{G_2}$, we have $\frac{|b_m|_{G_1}}{|b|_{G_1}}\geq \frac{|b_m|_{G_2}}{|b|_{G_2}}$.
%Note here, that if for all players $k \ne i$ we have $|b_k|_{G_1} = |b_k|_{G_2}$, then $j=i$. Hence (i) implies (ii).

Let $k$ denote the length of a shortest $s$-$t$ path in $G_1$, and $\ell$ the length of a shortest $s$-$t$ path in $G_2$. In the following we will observe that $\pi_j(b)$ only depends on $|b_j|_{G_1}$ and $|b_j|_{G_2}$ and not on the particular investments on shortest paths.
\begin{align*}
&\pi_i(b)\\
&=  \sum_{e \in E^+} b_{i,e} \cdot\frac{f_e^2}{b_e^2} - \sum_{e\in E^+}b_{i,e}\\
&=  \sum_{e \in E^+(G_1)} b_{i,e} \cdot\frac{f_e^2}{b_e^2}+ \sum_{e \in E^+(G_2)} b_{i,e} \cdot\frac{f_e^2}{b_e^2} - \sum_{e\in E^+}b_{i,e}\\
&=  \sum_{P \in \Pa(G_1) : b^1_{i,P}>0}\sum_{e \in P} b^1_{i,P} \cdot\frac{f_e^2}{b_e^2}+ \sum_{P \in \Pa(G_2) : b^2_{i,P}>0}\sum_{e \in P} b^2_{i,P} \cdot\frac{f_e^2}{b_e^2} - \sum_{e\in E^+}b_{i,e}\\
&=  \sum_{P \in \Pa(G_1) : b^1_{i,P}>0}b^1_{i,P} \cdot\sum_{e \in P} \left(\frac{|f|}{|b|_{G_1}}\right)^2+ \sum_{P \in \Pa(G_2) : b^2_{i,P}>0}b^2_{i,P} \cdot\sum_{e \in P} \left(\frac{|f|}{|b|_{G_2}}\right)^2 - \sum_{e\in E^+}b_{i,e}\\
&=k\cdot \left(\frac{|f|}{|b|_{G_1}}\right)^2\cdot|b_i|_{G_1} + \ell\cdot \left(\frac{|f|}{|b|_{G_2}}\right)^2\cdot |b_i|_{G_2} - \left(k\cdot |b_i|_{G_1}+\ell\cdot |b_i|_{G_2}\right),
\end{align*}
where the fourth equality follows because the Wardrop flow in $G$ can be decomposed into a Wardrop flow in $G_1$ and $G_2$ as all providers are investing in shortest paths.

Assume that provider $j$ chooses a vector of investments $b'_j$ with $|b'_j|_{G_1}=|b_j|_{G_1}+\epsilon$ and $|b'_j|_{G_2}=|b_j|_{G_2}-\frac{|b|_{G_2}^2\cdot\epsilon\cdot k}{|b|_{G_1}^2\cdot \ell+|b|_{G_1}\cdot\epsilon\cdot \ell+|b|_{G_2}\cdot\epsilon\cdot k}$, where $\epsilon>0$ is sufficiently small. We observe that $|f(b)| = |f(b')|$, since
\begin{align*}
c_b(|f|) &=\sum_{e \in P\cap E(G_1)} \frac{|f|}{|b|_{G_1}}+\sum_{e \in P\cap E(G_2)}\frac{|f|}{|b|_{G_2}}=k\cdot\frac{|f|}{|b|_{G_1}}+\ell\cdot\frac{|f|}{|b|_{G_2}}\\
&=k\cdot\frac{|f|}{|b|_{G_1}+\epsilon}+\ell\cdot\frac{|f|}{|b|_{G_2}-\frac{|b|_{G_2}^2\cdot\epsilon\cdot k}{|b|_{G_1}^2\cdot \ell+|b|_{G_1}\cdot\epsilon\cdot \ell+|b|_{G_2}\cdot\epsilon\cdot k}}=c_{b'}(|f|),
\end{align*}
where the third equality follows from basic calculus.

The total difference in the profit of provider $j$ is
\begin{align*}
& \pi_j(b'_j,b_{-j}) - \pi_j(b_j,b_{-j})\\
=&k\cdot\left(\left(\frac{|f|}{|b|_{G_1}+\epsilon}\right)^2\cdot(|b_j|_{G_1}+\epsilon)-\left(\frac{|f|}{|b|_{G_1}}\right)^2\cdot |b_j|_{G_1}\right)\\
+&\ell\cdot\left(\left(\frac{|f|}{|b|_{G_2}-\frac{|b|_{G_2}^2\cdot\epsilon\cdot k}{|b|_{G_1}^2\cdot \ell+|b|_{G_1}\cdot\epsilon\cdot \ell+|b|_{G_2}\cdot\epsilon\cdot k}}\right)^2\cdot\left(|b_j|_{G_2}-\frac{|b|_{G_2}^2\cdot\epsilon\cdot k}{|b|_{G_1}^2\cdot \ell+|b|_{G_1}\cdot\epsilon\cdot \ell+|b|_{G_2}\cdot\epsilon\cdot k}\right)\right.\\
-&\left.\left(\frac{|f|}{|b|_{G_2}}\right)^2\cdot |b_j|_{G_2}\right)-\left(k\cdot \epsilon-\ell \cdot \frac{|b|_{G_2}^2\cdot\epsilon\cdot k}{|b|_{G_1}^2\cdot \ell+|b|_{G_1}\cdot\epsilon\cdot \ell+|b|_{G_2}\cdot\epsilon\cdot k}\right)
\end{align*}
and thus
\begin{align*}
&\left.\frac{\partial \left(\pi_j(b'_j,b_{-j}) - \pi_j(b_j,b_{-j})\right)}{\partial \epsilon}\right|_{\epsilon=0}\\
&=\frac{k}{|b|_{G_1}^2}\cdot\left(2\cdot|f|^2\cdot\left(\frac{|b_j|_{G_2}}{|b|_{G_2}}-\frac{|b_j|_{G_1}}{|b|_{G_1}}\right)+\left(|b|_{G_2}^2-|b|_{G_1}^2\right)\right)>0,
\end{align*}
where the inequality follows because either $|b|_{G_1}=|b|_{G_2}$ and then $\frac{|b_j|_{G_1}}{|b|_{G_1}}<\frac{|b_j|_{G_2}}{|b|_{G_2}}$, or $|b|_{G_1}<|b|_{G_2}$ and $\frac{|b_j|_{G_1}}{|b|_{G_1}}\leq \frac{|b_j|_{G_2}}{|b|_{G_2}}$,
which completes the proof of (i).

In order to show (ii) we revisit the proof of (i) and use the additional property that $|b|_{G_1} - |b_i|_{G_1} = |b|_{G_2} - |b_i|_{G_2} \eqqcolon \alpha$. From (i) we conclude 
\begin{align*}
\pi_i(b) &= k\cdot \left(\frac{|f|}{|b|_{G_1}}\right)^2\cdot|b_i|_{G_1} + \ell\cdot \left(\frac{|f|}{|b|_{G_2}}\right)^2\cdot |b_i|_{G_2} - \left(k\cdot |b_i|_{G_1}+\ell\cdot |b_i|_{G_2}\right)\\
&= |f|^2 \left(\frac{k}{(\alpha + |b_i|_{G_1})^2}|b_i|_{G_1} + \frac{\ell}{(\alpha + |b_i|_{G_2})^2}|b_i|_{G_2}\right) - \left(k\cdot |b_i|_{G_1}+\ell\cdot |b_i|_{G_2}\right).
\end{align*}
In order to construct a path strategy on shortest paths, we choose a shortest path strategy $b_i'$ with $|b_i'|_{G_1} = |b_i'|_{G_2} = \beta$ for \[\beta = \frac{k |b_i|_{G_1} \left(\alpha + |b_i|_{G_2}\right) + |b_i|_{G_2} \ell \left(\alpha + |b_i|_{G_1}\right)}{k \left(\alpha + |b_i|_{G_2}\right) + \ell \left(\alpha + |b_i|_{G_1}\right)}.\] We can check easily that $\beta$ is exactly chosen in a way that
\[\frac{k}{\alpha + |b_i|_{G_1}} + \frac{\ell}{\alpha + |b_i|_{G_2}}=\frac{k+\ell}{\alpha + \beta},\]
i.e., such that $c_b(|f|) = c_{b_{-i},b_i'}(|f|)$, thus the switch of provider $i$ from $b_i$ to $b_i'$ is demand preserving. It remains to show that the profit is not decreasing. We observe
\begin{align*}
&\pi_i(b_{-i},b_i') - \pi_i(b)\\
&= |f|^2 \left(\frac{k+\ell}{(\alpha + \beta)^2}\beta  - \frac{k}{(\alpha + |b_i|_{G_1})^2}|b_i|_{G_1} - \frac{\ell}{(\alpha + |b_i|_{G_2})^2}|b_i|_{G_2}\right)\\
& \quad + \left(k\cdot |b_i|_{G_1}+\ell\cdot |b_i|_{G_2}\right) - \left(k + \ell\right) \beta\\
&= |f|^2 \frac{k \ell \alpha (|b_i|_{G_1} - |b_i|_{G_2})^2}{(k + \ell)^2 (\alpha + |b_i|_{G_1})^2 (\alpha + |b_i|_{G_2})^2} + \frac{k \ell (|b_i|_{G_1} - |b_i|_{G_2})^2}{|b_i|_{G_1} \ell + |b_i|_{G_2} k + \alpha (k + \ell)} \geq 0,
\end{align*} 
which completes the proof of (ii).
\end{proof}

\section*{Homogeneous Users}
\lemmaprofix*
\begin{proof}
Suppose that $|b|\leq\frac{d\cdot k}{R}$. Then the following equivalence follows $u(\delta(|b|))=k\cdot\frac{\delta(|b|)}{|b|}\Leftrightarrow R=k\cdot\frac{\delta(|b|)}{|b|}\Leftrightarrow \delta(|b|)=\frac{R\cdot |b|}{k}$. This implies that $\delta(|b|)=\min\left\{\frac{R\cdot |b|}{k},d\right\}$.

Using the definition of a provider's profit, we also get that $\pi_i(b_i,b_{-i})=k\cdot\left(\frac{\delta(|b|)^2}{|b|^2}-1\right)\cdot b_i=\min\left\{k\cdot\left(\frac{R^2}{k^2}-1\right)\cdot b_i,k\cdot\left(\frac{d^2}{|b|^2}-1\right)\cdot b_i\right\}$.
\end{proof}

\FixedExistenceOfNash*
\begin{proof}
By Lemma \ref{lem:prof},
\[
\pi_i(b_i,b_{-i})=\begin{cases}
k\cdot\left(\frac{R^2}{k^2}-1\right)\cdot b_i&\text{ if }|b|\leq\frac{d\cdot k}{R},\\
k\cdot\left(\frac{d^2}{|b|^2}-1\right)\cdot b_i&\text{ otherwise.}
\end{cases}
\]
We distinguish the following three cases: $R< k$, $R=k$, and $R>k$.

Suppose that $R< k$. Then $\frac{R^2}{k^2}< 1$ and thus $\pi_i(b)<0$ for all $i\in N$ with $b_i>0$ and all $b\in\mathbb{R}_+^N$ . Hence $b_i = 0$ for all $i\in N$ is the unique subgame perfect Nash equilibrium.

Now, suppose that $R=k$. Then $\frac{R^2}{k^2}=1$ and thus $\pi_i(b)=0$ for all $i\in N$ and all $b\in\mathbb{R}_+^N$ with $|b|\leq\frac{d\cdot k}{R}$, and $\pi_i(b)<0$ for all $i\in N$ with $b_i>0$ and all $b\in\mathbb{R}_+^N$ with $|b|>\frac{d\cdot k}{R}$. Hence all $b\in\mathbb{R}_+^N$ with $|b|\leq\frac{d\cdot k}{R}$ are a subgame perfect Nash equilibrium.

Finally, suppose that $R>k$. Assume that $b\in\mathbb{R}_+^N$ is a subgame perfect Nash equilibrium, we will derive some properties. First, since $\frac{R^2}{k^2}>1$, there cannot be a subgame perfect Nash equilibrium with $|b|<\frac{d\cdot k}{R}$, as the profit function of all $i\in N$ is a strictly increasing linear function in that region, and thus we can conclude that $|b|\geq\frac{d\cdot k}{R}$. Second, since $\pi_i(b)\geq 0$ for all $i\in N$, we must have $|b|\leq d$. In particular, this implies that for all $i\in N$ with $b_i>0$, $\pi_i(b)>0$ as by slightly decreasing the investment, the provider would gain a strictly positive profit. But then $\pi_i(b)>0$ for all $i\in N$, as $\left(\frac{d^2}{|b|^2}-1\right)>0$ and a provider without investment could gain a strictly positive profit by investing an arbitrary small, but positive amount. This implies that it is sufficient to check the first order conditions of all providers $i\in N$. Given the kink in the profit function at $|b|=\frac{d\cdot k}{R}$, the first order condition is given by $d^2\cdot\frac{|b|-2b_i}{|b|^3}\leq 1$ for all $i\in N$, where the inequality is an equality if $|b|>\frac{d\cdot k}{R}$. We now consider two subcases: either  $d^2\cdot\frac{|b|-2b_i}{|b|^3}< 1$ for some $i\in N$, or $d^2\cdot\frac{|b|-2b_i}{|b|^3}=1$ for all $i\in N$.

First, suppose that $d^2\cdot\frac{|b|-2b_i}{|b|^3}< 1$ for some $i\in N$. Then by the first order conditions, $|b|=\frac{d\cdot k}{R}$ and $b_i\geq d\cdot\frac{k\cdot(R^2-k^2)}{2R^3}$ for all $i\in N$. Summing over all $i\in N$, yields $\frac{k^2}{R^2}>\frac{n-2}{n}$. In particular, $b_i=\frac{d\cdot k}{n\cdot R}$ for all $i\in N$ is a subgame perfect Nash equilibrium, because $|b|=\frac{d\cdot k}{R}$ and $d^2\cdot\frac{|b|-2b_i}{|b|^3}=d^2\cdot\frac{\frac{n-2}{n}\cdot \left(\frac{d\cdot k}{R}\right)}{\left(\frac{d\cdot k}{R}\right)^3}<1$.

Second, suppose that $d^2\cdot\frac{|b|-2b_i}{|b|^3}= 1$ for all $i\in N$. Then by rewriting the first order condition, we have $b_i=d\cdot\sqrt{\frac{n-2}{n^3}}$ for all $i\in N$. This is a subgame perfect Nash equilibrium if $|b|\geq\frac{d\cdot k}{R}$, or equivalently $\sqrt{\frac{n-2}{n}}\geq\frac{k}{R}$. Notice that in this case the subgame perfect Nash equilibrium is unique.

\end{proof}

\section*{Heterogeneous Users }

\begin{lemma}
\label{lem:dem}
\label{lem:pro}
 Let $k$ denote the number of edges on a shortest $s$-$t$ path in $G$. For $u(x)=\frac{1}{x^{1/\alpha}}$, where $\alpha>0$, 
  \begin{itemize}
 \item[(i)]  the demand $|f(b)|$ is given by the function
$|f(b)|=\left(\frac{|b|}{k}\right)^{\frac{\alpha}{\alpha+1}}$ and
\item[(ii)] the profit $\pi_i$ is given by $
\pi_i(b_i, b_{-i})=b_i\left((b_i+|b|_{-i})^{\frac{-2}{\alpha+1}}{k}^{\frac{-\alpha+1}{\alpha+1}}-k \right).$
 \end{itemize}
\end{lemma}

\begin{proof}
%[Proof of Lemma \ref{lem:dem}]

 We  prove the two parts separately.
 \begin{itemize}
 \item[(i)] $u(x)=k\cdot\frac{x}{|b|}\Leftrightarrow\frac{1}{x^{1/\alpha}}=k\cdot\frac{x}{|b|}\Leftrightarrow x=\left(\frac{|b|}{k}\right)^{\frac{\alpha}{\alpha+1}}$.
\item[(ii)] $\pi_i(b_i,b_{-i})=k\left(|f(b_i+|b|_{-i})|^2 \frac{b_i}{(b_i+|b|_{-i})^2} - b_i\right)= b_i\left((b_i+|b|_{-i})^{\frac{-2}{\alpha+1}}{k}^{\frac{-\alpha+1}{\alpha+1}}-k \right).$
 \end{itemize}
\end{proof}

\begin{lemma}\label{lem:con}
Let $i\in N$. If there exists a provider $j\neq i$ with $b_j>0$ or $\alpha>1$, then $\pi_i(b_i,b_{-i})$ is continuous in $b_i$.
\end{lemma}

\begin{proof}
By Lemma \ref{lem:pro}, the profit function is a continuous function on $(0,\infty)$. So the only possible point of discontinuity is at $0$. Notice that $\pi_i(0,b_{-i})=0$ by definition. Remains to show that $\lim\limits_{b_i\rightarrow 0} \pi_i(b_i,b_{-i})=0$.

\noindent {\em Case 1}: there exists a provider $j\neq i$ with $b_j>0$. Then as $b_i\rightarrow 0$, the quantity in the bracket converges to the constant $\frac{k^{\frac{1-\alpha}{\alpha+1}}}{(|b|_{-i})^{\frac{2}{\alpha+1}}}-k$ and as $b_i\rightarrow 0$, we have $\lim\limits_{b_i\rightarrow 0} \pi_i(b_i,b_{-i})=0$.

\noindent {\em Case 2}: $b_j=0$ for all $j\neq i$ and $\alpha>1$. Then by Lemma \ref{lem:pro}, we have
$b_i^{\frac{\alpha-1}{\alpha+1}}{k}^{\frac{-\alpha+1}{\alpha+1}}-k\cdot b_i.$
Since $\alpha>1$, $\lim\limits_{b_i\rightarrow 0} b_i^{\frac{\alpha-1}{\alpha+1}}=0$ and hence $\lim\limits_{b_i\rightarrow 0}\pi_i(b_i,b_{-i})=0$. 
\end{proof}

\theoremOne*

\begin{proof}
Assume $n=1$. By Lemma \ref{lem:pro} the payoff for one provider evaluates to
$\pi_1(b_1)=b_1\left(b_1^{\frac{-2}{\alpha+1}}k^{\frac{1-\alpha}{\alpha+1}}-k\right)$.
We consider three different cases.

\noindent {\em Case 1}: $\alpha<1$. Then $\lim\limits_{b_1\rightarrow 0}\pi_1(b_1)=\infty$ and thus there is no best response.

\noindent {\em Case 2}: $\alpha=1$. Then $\pi_1(b_1)=1-k\cdot b_1$ and and thus there is no best response.

\noindent {\em Case 3}: $\alpha>1$. By Lemma \ref{lem:con}, the profit function is continuous. Moreover, we have $\lim\limits_{b_1\rightarrow \infty}\pi_1(b_1)=-\infty$. This implies that any Nash equilibrium with positive profit for provider $1$ has to fulfill the first-order condition: % for all $i\in S$:
$
\frac{\partial\pi_1(b_1)}{\partial b_1}=\frac{\alpha-1}{\alpha+1}b_1^{\frac{-2}{\alpha+1}}k^{\frac{1-\alpha}{\alpha+1}}-k=0,
$
yielding
$b_1=\frac{\left(1-\frac{2}{(\alpha+1)}\right)^{\frac{\alpha+1}{2}}}{k^\alpha}.$
Since this is the unique point satisfying the first order conditions, this must be a global maximum.
\end{proof}

\propositionInv*

\begin{proof}
%[Proof of Proposition \ref{pro:inv}]
If $|S|=1$, see Theorem \ref{thm:one}. Assume that $|S|\geq 2$. By Lemma \ref{lem:con}, the profit function is continuous. Moreover, $\lim\limits_{b_i\rightarrow \infty}(b_i+|b|_{-i})^{\frac{-2}{\alpha+1}}k^{\frac{1-\alpha}{\alpha+1}}=0$ and hence $\lim\limits_{b_i\rightarrow \infty}\pi_i(b_i,b_{-i})=-\infty$. This implies that any Nash equilibrium with positive profit for provider $i$ has to fulfill the first-order condition for all $i\in S$:
\begin{align*}
&\frac{\partial \pi_i(b_i,b_{-i})}{\partial b_i} = 0\\
\Leftrightarrow\;\;\;&\frac{k^{\frac{1-\alpha}{\alpha+1}}}{(b_i+|b|_{-i})^{\frac{2}{\alpha+1}}}\cdot\left(1-\frac{2}{(\alpha+1)}\cdot\frac{b_i}{b_i+|b|_{-i}}\right)-k=0\\
\Leftrightarrow\;\;\;&b_i=\frac{k^{\frac{1-\alpha}{\alpha+1}}\cdot |b|^{\frac{-2}{\alpha+1}}-k}{\frac{2}{\alpha+1}\cdot k^{\frac{1-\alpha}{\alpha+1}}\cdot |b|^{\frac{-\alpha-3}{\alpha+1}}}\;.
\end{align*}
The right-hand side is independent on $i$ and the same for all providers $i\in S$. Hence $b_i=b_j$ for all $i\in S$. Setting $b_i=\frac{1}{|S|}\cdot |b|$ yields $b_i=\frac{1}{|S|}\cdot\frac{\left(1-\frac{2}{(\alpha+1)\cdot |S|}\right)^{\frac{\alpha+1}{2}}}{k^\alpha}$.
\end{proof}

\propositionAllinvest*
Before we can prove Lemma \ref{lem:all}, we need the following two lemmas.

\begin{lemma}\label{lem:zer}
$b_i=0$ for all $i\in N$ is not a Nash equilibrium.
\end{lemma}
\begin{proof}
Let $i\in N$ and assume that $b_j=0$ for all $j\neq i$. Suppose that $b_i=0$ is a best response. We derive a contradiction.
\[\pi_i(b_i,b_{-i})=k^{\frac{1-\alpha}{\alpha+1}}\cdot b_i^{\frac{\alpha-1}{\alpha+1}}-k\cdot b_i>0\Leftrightarrow b_i<\frac{1}{k^{\alpha}}.\]
Hence there always exists a $b_i>0$ so that $\pi_i(b_i,b_{-i})>0$.
\end{proof}

\begin{lemma}\label{lem:one}
$b_i>0$ for some $i\in N$ and $b_j=0$ for all $j\neq i$ is not a Nash equilibrium.
\end{lemma}
\begin{proof}
Let $i\in N$ and assume that $b_j=0$ for all $j\neq i$. Suppose that $b_i>0$ is a best response. We derive a contradiction. Recall that $\pi_i(b_i,b_{-i})=k^{\frac{1-\alpha}{\alpha+1}}\cdot b_i^{\frac{\alpha-1}{\alpha+1}}-k\cdot b_i$.

Case 1: $\alpha<1$. Then $\lim\limits_{b_i\rightarrow 0}\pi_i(b_i,b_{-i})=\infty$ and hence $0<b'_i<b_i$ is a better response.

Case 2: $\alpha=1$. Then $\pi_i(b_i,b_{-i})=1-k\cdot b_i$ and hence $0<b'_i<b_i$ is a better response.

Case 3: $\alpha>1$. By Theorem \ref{thm:one}, $b_i=\frac{\left(1-\frac{2}{(\alpha+1)}\right)^{\frac{\alpha+1}{2}}}{k^\alpha}$. Since the profit function of provider $j\neq i$ is continuous by Lemma \ref{lem:con}, it is sufficient to show that $\frac{\partial \pi_j(b_j,b_{-j})}{\partial b_j}>0$ at $b_j=0$. Hence $\frac{k^{\frac{1-\alpha}{\alpha+1}}}{b_i^{\frac{2}{\alpha+1}}}-k=k\cdot\left(1-\frac{2}{(\alpha+1)}\right)^{-1}-k>0$.
\end{proof}

\begin{proof}[Proof of Lemma \ref{lem:all}]
If $|S|=0$, see Lemma \ref{lem:zer}. If $|S|=1$, see Lemma \ref{lem:one}. Suppose that $|S|\geq 2$. By Lemma \ref{lem:inv}, $|b|=\frac{\left(1-\frac{2}{(\alpha+1)\cdot |S|}\right)^{\frac{\alpha+1}{2}}}{k^\alpha}$.
Since the profit function of provider $i\in N\setminus S$ is continuous by Lemma \ref{lem:con}, it is sufficient to show that $\frac{\partial \pi_i(b_i,b_{-i})}{\partial b_i}>0$ at $b_i=0$. Evaluating $\frac{\partial \pi_i(b_i,b_{-i})}{\partial b_i}$ at $b_i=0$ yields $\frac{k^{\frac{1-\alpha}{\alpha+1}}}{|b|^{\frac{2}{\alpha+1}}}-k$. Plugging in $|b|$ yields $k\cdot\left(1-\frac{2}{(\alpha+1)\cdot |S|}\right)^{-1}-k = l$, which is clearly positive since $\left(1-\frac{2}{(\alpha+1)\cdot |S|}\right)^{-1}>1$.
\end{proof}

\PoAElastic*
\begin{proof}
We begin the proof by deriving a closed form expression for the social welfare.
\[SW(b)=\int_0^{|f(b)|}u(x)\;dx-k |b|=\int_0^{\left(\frac{|b|}{k}\right)^{\frac{\alpha}{\alpha+1}}}x^{-\frac{1}{\alpha}}\;dx-k |b|=\frac{\alpha}{\alpha-1}\left(\frac{|b|}{k}\right)^{\frac{\alpha}{\alpha+1}}-k |b|,\] where we used $|f(b)|=\left(\frac{|b|}{k}\right)^{\frac{\alpha}{\alpha+1}}$ which is proven in Lemma \ref{lem:dem} in the appendix.

In the following, we can compute the social optimum. Note that the social welfare function is continuous in $|b|$ and $\lim_{b\rightarrow\infty}SW(b)=-\infty$. Hence the first order conditions, if unique, are sufficient to ensure a global maximum,
%\begin{align*}
$
\frac{\alpha}{\alpha+1}|b|^{\frac{-2}{\alpha+1}}k^{\frac{1-\alpha}{\alpha+1}}-k=0\;
\Leftrightarrow\; |b|=\left(\frac{\alpha}{\alpha+1}\right)^{\frac{\alpha+1}{2}}k^{-\alpha}.$
%\end{align*}
\noindent Thus, the optimal social welfare is

%\begin{align*}
$\frac{\alpha}{\alpha-1}\left(\frac{\left(\frac{\alpha}{\alpha+1}\right)^{\frac{\alpha+1}{2}}k^{-\alpha}}{k}\right)^{\frac{\alpha-1}{\alpha+1}}-k\left(\frac{\alpha}{\alpha+1}\right)^{\frac{\alpha+1}{2}}k^{-\alpha}=\left(\frac{\alpha}{\alpha+1}\right)^{\frac{\alpha-1}{2}}k^{1-\alpha}\frac{2\alpha}{(\alpha-1)(\alpha+1)}.
$
%\end{align*}

\noindent  By Theorem \ref{thm:one}, the equilibrium social welfare is

%\begin{align*}
$
\frac{\alpha}{\alpha-1}\left(\frac{\frac{\left(1-\frac{2}{(\alpha+1)\cdot |N|}\right)^{\frac{\alpha+1}{2}}}{k^\alpha}}{k}\right)^{\frac{\alpha-1}{\alpha+1}}-k \frac{\left(1-\frac{2}{(\alpha+1)\cdot |N|}\right)^{\frac{\alpha+1}{2}}}{k^\alpha}
=\left(1-\frac{2}{(\alpha+1)n}\right)^{\frac{\alpha-1}{2}}k^{1-\alpha}\left(\frac{1}{\alpha-1}+\frac{2}{(\alpha+1)n}\right).
$%\end{align*}

\noindent  Thus,
%\begin{align*}
$
PoA(\mathcal{G})=\frac{\left(\frac{\alpha}{\alpha+1}\right)^{\frac{\alpha-1}{2}}k^{1-\alpha}\frac{2\alpha}{(\alpha-1)(\alpha+1)}}{\left(1-\frac{2}{(\alpha+1)n}\right)^{\frac{\alpha-1}{2}}k^{1-\alpha}\left(\frac{1}{\alpha-1}+\frac{2}{(\alpha+1)n}\right)}
=\left(\frac{\alpha n}{n(\alpha+1)-2}\right)^{\frac{\alpha-1}{2}}\frac{2\alpha n}{n(\alpha+1)+2(\alpha-1)}.
$

%\end{align*}
\noindent  To prove the last inequality, observe that if $n=1$,
$PoA(\mathcal{G})=\frac{2\alpha\left(\frac{\alpha}{\alpha-1}\right)^{\frac{\alpha-1}{2}}}{3\alpha-1}<2\sqrt{\frac{1}{e}}.$

\noindent  For $n\geq2$, we have
$\frac{\partial PoA(\mathcal{G})}{\partial n}=\frac{2(\alpha^2-1)(n-1)\left(\frac{\alpha n}{n(\alpha+1)-2}\right)^{\frac{\alpha+1}{2}}}{n(n(\alpha+1)+2(\alpha-1))^2}\geq0.$

\noindent As $n\rightarrow\infty$, $PoA(\mathcal{G})\rightarrow 2\left(\frac{\alpha}{\alpha+1}\right)^{\frac{\alpha-1}{2}}\frac{\alpha}{\alpha+1}=2\sqrt{\left(1-\frac{1}{\alpha+1}\right)^{\alpha+1}}\leq 2\sqrt{\frac{1}{e}}.$
\end{proof}

\theoremPPOA*

For the proof of the theorem, the following simple lemma on the providers' profit in equilibrium is crucial. 

\begin{lemma}\label{lem:pro2}
%For $u(x)=\frac{1}{x^{1/\alpha}}$ with $\alpha>1$, or $\alpha>0$ and $n\geq 2$,
If $b\in\mathcal{E}(\mathcal{G})$, then 
$\sum_{i\in N}\pi_i(b)=\left(1-\frac{2}{(\alpha+1)\cdot |N|}\right)^{\frac{\alpha-1}{2}}k^{1-\alpha}\frac{2}{(\alpha+1)n}.$
\end{lemma}
\begin{proof}
By Theorem \ref{thm:one},
$|b|=\frac{\left(1-\frac{2}{(\alpha+1)\cdot |N|}\right)^{\frac{\alpha+1}{2}}}{k^\alpha}.$
By Lemma \ref{lem:pro},
$\sum_{i\in N}\pi_i(b)=|b|^{\frac{\alpha-1}{\alpha+1}}k^{\frac{1-\alpha}{\alpha+1}}-|b|\cdot k=\left(1-\frac{2}{(\alpha+1)\cdot |N|}\right)^{\frac{\alpha-1}{2}}k^{1-\alpha}\frac{2}{(\alpha+1)n}.$
\end{proof}
\begin{proof}[Proof of Theorem \ref{thm:ppoa2}]
We consider three different cases.
\noindent {\em Case 1}: $\alpha<1$. Since $\sup\limits_b \sum_{i\in N}\pi_i(b)\rightarrow\infty$ as $|b|\rightarrow 0$ and $\min\limits_{b\in \mathcal{E}(\mathcal{G})}\sum_{i\in N}\pi_i(b)<\infty$, we conclude the result.

\noindent {\em Case 2}: $\alpha=1$. Since $\sup\limits_b \sum_{i\in N}\pi_i(b)<1$ for all $|b|\geq 0$, we have by Lemma \ref{lem:pro2} that $PPoA(\mathcal{G})=n.$

\noindent {\em Case 3}: $\alpha>1$. By Lemma \ref{lem:pro2},
%\begin{align*}
$PPoA(\mathcal{G})=\frac{\left(1-\frac{2}{(\alpha+1)}\right)^{\frac{\alpha-1}{2}}k^{1-\alpha}\frac{2}{(\alpha+1)}}{\left(1-\frac{2}{(\alpha+1)\cdot |N|}\right)^{\frac{\alpha-1}{2}}k^{1-\alpha}\frac{2}{(\alpha+1)n}}=n\left(\frac{(\alpha+1)n-2}{n(\alpha-1)}\right)^{\frac{1-\alpha}{2}}.$
%\end{align*}
\end{proof}
\end{document}